\documentclass[10pt,a4paper, UKenglish]{article}
\usepackage{fullpage}
\usepackage[utf8]{inputenc}
\usepackage{graphicx}

\RequirePackage{hyperref}
\hypersetup{colorlinks=true, urlcolor=Blue, citecolor=Green, linkcolor=BrickRed, breaklinks, unicode}

\usepackage[dvipsnames,usenames]{xcolor}
\usepackage[nocompress]{cite}
\usepackage{amsmath,mathtools,stmaryrd}

 \newtheorem{lemma}{Lemma}
 \newtheorem{theorem}{Theorem}
 \newtheorem{corollary}{Corollary}

\newenvironment{proof}{\noindent\emph{Proof. }}

\newcommand{\spine}{\ensuremath\mathrm{spine}}
\newcommand{\locus}{\ensuremath\mathrm{locus}}
\newcommand{\topboundary}{\ensuremath\mathrm{top}}
\newcommand{\height}{\ensuremath\mathrm{height}}
\newcommand{\depth}{\ensuremath\mathrm{depth}}

\newcommand{\rank}{\ensuremath\mathrm{rank}}

\newcommand{\parent}{\ensuremath\mathrm{parent}}
\newcommand{\query}{\ensuremath\mathrm{query}}
\newcommand{\checkptr}{\ensuremath\mathrm{check}}
\newcommand{\jumpslice}{\ensuremath\mathrm{next}}
\newcommand{\jumptop}{\ensuremath\mathrm{top}}
\newcommand{\jumpbottom}{\ensuremath\mathrm{bottom}}

\newcommand{\hexit}{\ensuremath\mathrm{hexit}}
\newcommand{\ventry}{\ensuremath\mathrm{ventry}}
\newcommand{\vexit}{\ensuremath\mathrm{vexit}}
\newcommand{\hentry}{\ensuremath\mathrm{hentry}}

\newcommand{\T}{\ensuremath{\mathcal{T}}}

\newcommand{\TD}{\ensuremath{\mathcal{T\!D}}}
\newcommand{\V}{\ensuremath{\mathcal{V}}}
\newcommand{\VD}{\ensuremath{\mathcal{V\!D}}}
\renewcommand{\H}{\ensuremath{\mathcal{H}}}
\newcommand{\HD}{\ensuremath{\mathcal{H\!D}}}

\renewcommand{\S}{\ensuremath{\mathcal{S}}}
\renewcommand{\G}{\ensuremath{\mathcal{G}}}

\newcommand{\fp}{\ensuremath\phi}

\title{Top Tree Compression of Tries\footnote{An extended abstract appeared at ISAAC 2019~\cite{BGGLW2019}}}

\author{Philip Bille \\ \texttt{phbi@dtu.dk} \and 
Pawe\l{} Gawrychowski \\ \texttt{gawry@cs.uni.wroc.pl} \and 
Inge Li G{\o}rtz \\ \texttt{inge@dtu.dk} \and 
Gad M. Landau \\ \texttt{landau@cs.haifa.ac.il} \and 
Oren Weimann \\ \texttt{oren@cs.haifa.ac.il}}

\date{}

\begin{document}

\maketitle

\begin{abstract}
We present a compressed representation of tries based on top tree compression [ICALP 2013] that works on a standard, comparison-based, pointer machine model of computation and supports efficient prefix search queries. Namely, we show how to preprocess a set of strings of total length $n$ over an alphabet of size $\sigma$ into a compressed data structure of worst-case optimal size $O(n/\log_\sigma n)$ that given a pattern string $P$ of length $m$ determines if $P$ is a prefix of one of the strings in time $O(\min(m\log \sigma,m + \log n))$. We show that this query time is in fact optimal regardless of the size of the data structure.

Existing solutions either use $\Omega(n)$ space or rely on word RAM techniques, such as tabulation, hashing, address arithmetic, or word-level parallelism, and hence do not work on a pointer machine. Our result is the first solution on a pointer machine that achieves worst-case $o(n)$ space. Along the way, we develop several interesting data structures that work on a pointer machine and are of independent interest. These include an optimal data structures for random access to a grammar-compressed string and an optimal data structure for a variant of the level ancestor problem.
\end{abstract}

\section{Introduction}
A \emph{string dictionary} compactly represents a set of strings $S = S_1, \ldots, S_k$ to support efficient \emph{prefix queries}, that is, given a pattern string $P$ determine if $P$ is a prefix of some string in $S$. Designing efficient string dictionaries is a fundamental data structural problem dating back to the 1960's. String dictionaries are a key component in a wide range of applications in areas such as computational biology, data compression, data mining, information retrieval, natural language processing, and pattern matching. 

A key challenge and the focus of most of the recent work is to design efficient \emph{compressed string dictionaries}, that take advantage of repetitions in the strings to minimize space, while still supporting efficient queries. While many efficient solutions are known, they all rely on powerful word-RAM techniques, such as tabulation, address arithmetic, word-level parallelism, hashing, etc., to achieve efficient bounds. A natural question is whether or not such techniques are necessary for obtaining efficient compressed string dictionaries or if simpler and more basic computational primitives such as pointer-based data structures and character comparison suffice.

In this paper, we answer this question to the affirmative by introducing a new compressed string dictionary based on \emph{top tree compression} that works on a standard comparison-based, pointer machine model of computation. We achieve the following bounds: let $n = \sum_{i=1}^k |S_i|$ be the total length of the strings in $S$, let $\sigma$ be the size of the alphabet, and $m$ be the length of a query string $P$. Our compressed string dictionary uses $O(n/\log_\sigma n)$ space (space is measured as the number of words and not bits, see discussion below) and supports queries in $O(\min(m\log \sigma,m + \log n))$ time. The space matches the information-theoretic worst-case space lower bound, and we further show that the query time is optimal for any comparison-based query algorithm regardless of the space. Compared to previous work our string dictionary is the first $o(n)$ space solution in this model of computation.

\subsection{Computational Models} 
We consider three computational models. In the \emph{comparison-based model} algorithms only interact with the input by comparing elements. Hence they cannot exploit the internal representation of input elements, e.g., for hashing or word-level parallelism. The comparison-based model is a fundamental and well-studied computational model, e.g., in textbook results for sorting~\cite{Knuth1969}, string matching~\cite{KMP1977}, and computational geometry~\cite{PH1977}. Modern programming languages and libraries, such as the \texttt{C++} standard template library, implement comparison-based algorithms by supporting abstract and user-specified comparison functions as function arguments. In our context, we say that a string dictionary is comparison-based if the query algorithm can only access the input string $P$ via single character comparisons of the form $P[i] \leq c$, where $c$ is a character.	

In the \emph{pointer machine model}, a data structure is a directed graph with bounded out-degree. Each node contains a constant number of data fields or pointer to other nodes and algorithms must access the data structure by traversing the graph. Hence, a pointer machine algorithm cannot implement random access structures such as arrays or perform address arithmetic. The pointer machine captures linked data structures such as linked-lists and search trees. The pointer machine model is a classic and well-studied model, see e.g.~\cite{Ta79, Ch90, HMZ16, CR96, AAL12}. 

Finally, in the word RAM model of computation~\cite{Hagerup1998} the memory is an array of memory words, that each contain a logarithmic number of bits. Memory words can be operated on in unit-time using a standard set of arithmetic operations, boolean operations, and shifts. The word RAM model is strictly more powerful than the comparison-based model and the pointer-machine model and supports random access, hashing, address arithmetic, word-level parallelism, etc. (these are not possible in the other models).

The space of a data structure in the word RAM model is the number of memory words used and the space in the pointer machine model is the total number of nodes. To compare the space of the models, we assume that each field in a node in the pointer machine stores a logarithmic number of bits. Hence, the total number of bits we can represent in a given space in both models is within a constant factor of each other. 

\subsection{Previous work}
The classic textbook string dictionary solution, due to Fredkin~\cite{Fredkin} from 1960, is to store the \emph{trie} $T$ of the strings in $S$ and to answer prefix queries using a top-down traversal of $T$, where at each step we match a single character from $P$ to the labels of the outgoing edges of a node. If we manage to match all characters of $P$ then $P$ is a prefix of a string in $S$ and otherwise it is not.

Depending on the representation of the trie and the model of computation we can obtain several combinations of space and time complexity. On a comparison-based, pointer machine model of computation, we can store the outgoing edges of each in a biased search tree~\cite{BentSleatorTarjan85}, leading to an $O(n)$ space solution with query time $O(\min(m\log \sigma,m + \log n))$. 

We can compress this solution by merging maximal identical complete subtrees of $T$~\cite{DST1980}, thus replacing $T$ by a directed acyclic graph (DAG) $D$ that represents $T$. This leads to a solution with the same query time as above but using only $O(d)$ space, where $d$ is the size of the smallest DAG $D$ representing $T$. The size of $D$ can be exponentially smaller than $n$, but may not compress at all. Consider for instance the case where $T$ is a single path of length $n$ where all edges have the same label (i.e., corresponding to a single string of the same letter). Even though $T$ is highly compressible (we can represent it by the label and the length of the path) it does not contain any identical subtrees and hence its smallest DAG has size $\Omega(n)$.  

Using the power of the word RAM model improved representations are possible. Benoit et al.~\cite{BDMRVS2005} and Raman et al.~\cite{RRS2007} gave \emph{succinct} representations of tries that achieve $O(n/\log_\sigma n)$ space and $O(m)$ query time, thus simultaneously achieving optimal query time and matching the worst-case information theoretic space lower bounds. These results rely on powerful word RAM techniques to obtain the bounds, such as  tabulation and hashing. Numerous trie representations are known, see e.g., \cite{darragh1993bonsai,PR2015,GO2015,KMF2017,Aoe1989,AF2018,AS2010,BBV2010,AF2014,BGS2017,TKKNIBT2019,YK2014,KMF2017a,TISA2017,yata2011}, but these all use word RAM techniques to achieve near optimal combinations of time and space.

Another approach is to compress the strings according to various measures of repetitiveness, such as the empirical $k$-th order entropy~\cite{GV2005,Makinen2003,NM2007,Sadakane2000}, the size of the Lempel-Ziv parse~\cite{BCGPM2015,CE2018,Karkkainen95lempel-zivparsing,Gagie2012,gagie2014lz77,BEGV2018,TIIBT2019}, the size of the smallest grammar~\cite{CN2011,CN2012,Gagie2012}, the run-length encoded Burrows-Wheeler transform, \cite{MN2005,MNSV2009,MNSV2010,SVMN2008}, and others~\cite{NP2019,BGGKOPT2015,TGFIA2017,BGGMS2014,FGNPS2017,AF2014}. The above solutions are designed to support more general queries on the strings, but as noted by Ars and Fischer~\cite{AF2014} they are straightforward to adapt to prefix queries. For example, if $z$ is size of the Lempel-Ziv parse of the concatenation of the strings in $S$, the result of Christiansen and Etienne~\cite{CE2018} implies a string dictionary of size $O(z \log (n/z))$ that supports queries in time $O(m + \log^\epsilon n)$. Since $z$ can be  exponentially smaller than $n$, the space is significantly improved on highly-compressible strings. Since $z = O(n/\log_\sigma n)$ in the worst-case, the space is always $O(\frac{n}{\log_\sigma n} \log (\frac{n}{n/\log_\sigma n})) = O(\frac{n \log \log_\sigma n}{\log_\sigma n})$ and thus almost optimal compared to the information theoretic lower bound. Similar bounds are known for the other measures of repetitiveness. As in the case of succinct representations of tries, all of these solutions use word RAM techniques. 

\subsection{Our results}
We propose a new compressed string dictionary that achieves the following bounds:
\begin{theorem}\label{thm:main}
	Let $S$ be a set of strings of total length $n$ over an alphabet of size $\sigma$. On a comparison-based, pointer machine model of computation, we can construct a compressed string dictionary that uses $O(n/\log_\sigma n)$ space and answer queries in $O(\min(m\log \sigma,m + \log n))$ time.
\end{theorem}
Note that the space bound for Theorem~\ref{thm:main} matches the information theoretic lower bound and the time bound matches the classic linear space implementation of tries with biased search trees. The result is the first $o(n)$ space solution in this model of computation. Furthermore, we show that this time bound is optimal. 

\begin{theorem}
\label{thm:comparison}
For any $n$, $m\leq n$, and $\sigma \geq 2$, there exists a set $S$ of strings of total length $n$ over an alphabet of size $\sigma$ such that any comparison-based algorithm that checks if a given pattern $P$ of length $m$ belongs to $S$ needs to perform $\Omega(\min(m\log\sigma,m+\log n))$ comparisons in the worst case.
\end{theorem}
Note that Theorem~\ref{thm:comparison} holds regardless of the space used, holds even for weaker membership queries, and only assumes that the algorithm is a comparison-based algorithm. We note that the upper bound holds on a pointer machine with comparisons and additions as arithmetic operations, while the lower bound only assumes comparisons.  

\subsection{Techniques}
%



In {\em top tree compression}~\cite{BGLW2015} one transforms a labeled tree $T$ into another tree $\T$ (called a \emph{top tree}) that is of height $O(\log n)$ and represents a hierarchical decomposition of $T$ into connected subgraphs (called \emph{clusters}). Each cluster overlaps with other clusters in at most two nodes. Every leaf in $\T$ corresponds to a cluster consisting of a single edge in $T$ and every internal node in $\T$ corresponds to a merge of two clusters. The top tree $\T$ is then compressed using the classical DAG compression resulting in the \emph{top DAG} $\TD$. The top DAG supports basic navigational queries on $T$ in $O(\log n)$ time, has size $O(n/\log_\sigma n)$, can compress exponentially better than DAG compression, and is never worse than DAG compression by more than a $O(\log n)$ factor~\cite{HR2015,BGLW2015,BFG2017,DG2018}.


Our main technical contribution is implementing prefix search optimally on the top DAG. To this end, we develop several optimal pointer machine data structures of independent interest:

\begin{itemize} 
\item A data structure for the \emph{path extraction problem}, that asks to compactly represent an edge-labeled tree $T$ such that given a node $v$ we can efficiently return the labels on the root-to-$v$ path in $T$. While an optimal solution for this problem can be obtained by plugging in
known tools, more specifically a fully persistent queue~\cite{HoodM81}, we believe that our self-contained solution is simpler and elegant.

\item
A data structure for the \emph{weighted level ancestor} problem, that asks to compactly represent an edge-weighted tree $T$ such that given a node $v$ and a positive number $x$ we can efficiently return the rootmost ancestor of $v$ whose distance from the root is at least $x$. An immediate implication of our weighted level ancestor data structure is an optimal data structure for the \emph{random access problem} on grammar compressed strings
. This improves a SODA'11 result~\cite{BLRSW2015} that required word RAM bit tricks. 

\item  A data structure for the \emph{spine path extraction problem}, that asks to compactly represent a top-tree compression $\TD$ such that given a cluster $C$ we can efficiently return the characters of the unique path between the two boundary nodes of $C$. 

\item
For the lower bound, we show that any algorithm that given a string $P[1,m]$ checks if $\sum_{i=1}^{m} P[i] = 0 \pmod 2$ needs to perform $\Omega(m\log \sigma)$ comparisons in the worst case. We then show  that when $n\geq m\sigma^{m}$ this implies the $\Omega(m\log \sigma)$ bound for our problem and when $n< m\sigma^{m}$ it implies the $\Omega(m + \log n)$ bound for our problem.  
\end{itemize} 

\subsection{Roadmap}
In Section~\ref{sec:preliminaries} we recall \emph{top trees} and how a top tree of a tree $T$ is obtained by merging (either vertically or a horizontally) the top trees of two subtrees of $T$ that overlap on a single node. In Section~\ref{sec:simple} we present a simple randomized Monte-Carlo  word RAM solution to the compressed string indexing problem that is the basis of our deterministic pointer machine solutions in the following sections. The solution is based on top trees and efficiently handles horizontal merges (deterministically) and vertical merges (randomized Monte-Carlo). 
In Section~\ref{sec:spineextraction} we show how to handle vertical merges deterministically on a pointer machine, and in Section~\ref{sec:part1puttingtogether} we show that this suffices to achieve the $O(m + \log n)$ query time in Theorem~\ref{thm:main}. 
We show a different way to handle vertical merges in Section~\ref{sec:Spine Path Extraction} and horizontal merges in Section~\ref{sec:horizontalaccess}. In Section~\ref{sec:mlogsigma} we show that these suffice to achieve the $O(m \log \sigma)$ query time in  Theorem~\ref{thm:main}. Finally, in Section~\ref{sec:lowerbound} we give a matching lower bound showing that the query time in Theorem~\ref{thm:main} is optimal regardless of the size of the structure.

\section{Preliminaries}\label{sec:preliminaries}
 In this section we briefly review Karp-Rabin fingerprints~\cite{KarpRabin}, top trees~\cite{AHLT2005}, and top tree compression~\cite{BGLW2015}.  

\subsection{Karp-Rabin Fingerprints} 
 The Karp-Rabin fingerprint~\cite{KarpRabin} of a string $x$ is defined as $\fp(x) = \sum_{i=1}^{|x|} x[i] \cdot c^i \bmod p$, where $c$ is a randomly chosen positive integer, and $2N^{c+4} \leq p \leq 4N^{c+4}$ is a prime. Karp-Rabin fingerprints guarantee that given two strings $x$ and $y$, if $x = y$ then $\fp(x) = \fp(y)$. Furthermore, if $x \neq y$, then with high probability $\fp(x) \neq \fp(y)$. Fingerprints can be composed and subtracted as follows.
\begin{lemma}\label{lemma:fingerprints}
Let $x = y z$ be a string decomposable into a prefix $y$ and suffix $z$. Given any two of the Karp-Rabin fingerprints $\fp(x)$, $\fp(y)$ and $\fp(z)$, it is possible to calculate the remaining fingerprint in constant time.
\end{lemma}

\subsection{Clustering}  
Let $v$ be a node in $T$ with children $v_1, \ldots, v_k$ in left-to-right order. Define $T(v)$ to be the subtree induced by $v$ and all proper descendants of $v$. Define $F(v)$ to be the forest induced by all proper descendants of $v$. For $1\leq s \leq r \leq k$ let $T(v, v_s, v_r)$ be the connected component induced by the nodes $\{v\} \cup T(v_s) \cup T(v_{s+1}) \cup \cdots \cup T(v_r)$.

A {\em cluster} with {\em top boundary node} $v$ is a connected component of the form $T(v, v_s, v_r)$, $1\leq s \leq r \leq k$. A {\em cluster} with {\em top boundary node} $v$ and {\em bottom boundary node} $u$ is a connected component of the form $T(v, v_s, v_r) \setminus F(u)$, $1 \leq s \leq r \leq k$, where $u$ is a node in $ T(v_s) \cup  \cdots \cup T(v_r)$. We denote the top boundary node of a cluster $C$ by $\topboundary(C)$. 
Clusters can therefore have either one or two boundary nodes. For example, let $p(v)$ denote the parent of $v$ then a single edge 
$(v,p(v))$  of $T$ is a cluster where $p(v)$ is the top boundary node. If $v$ is a leaf then there is no bottom boundary node, otherwise $v$ is a bottom boundary node. Nodes that are not boundary nodes are called \emph{internal} nodes. The path between the top and bottom boundary nodes in a cluster $C$ is called the cluster's spine, and the string obtained by concatenating the labels on the spine from top to bottom is denoted $\spine(C)$. 

Two edge disjoint clusters $A$ and $B$ whose vertices overlap on a single boundary node can be \emph{merged} if their union $C = A \cup B$ is also a cluster. There are five ways of merging clusters (see Figure~\ref{fig:merge}). Merges of type (a) and (b) are called {\em vertical} merges ($C$ is then a {\em vertical} cluster) and can be done only if the common boundary node is not a boundary node of any other cluster except $A$ and $B$. Merges of type (c),(d), and (e) are called {\em horizontal} merges ($C$ is then a {\em horizontal} cluster) and can be done only if at least one of $A$ or $B$ does not have a bottom boundary node.

\begin{figure}[t]
   \centering
   \includegraphics[scale=0.30]{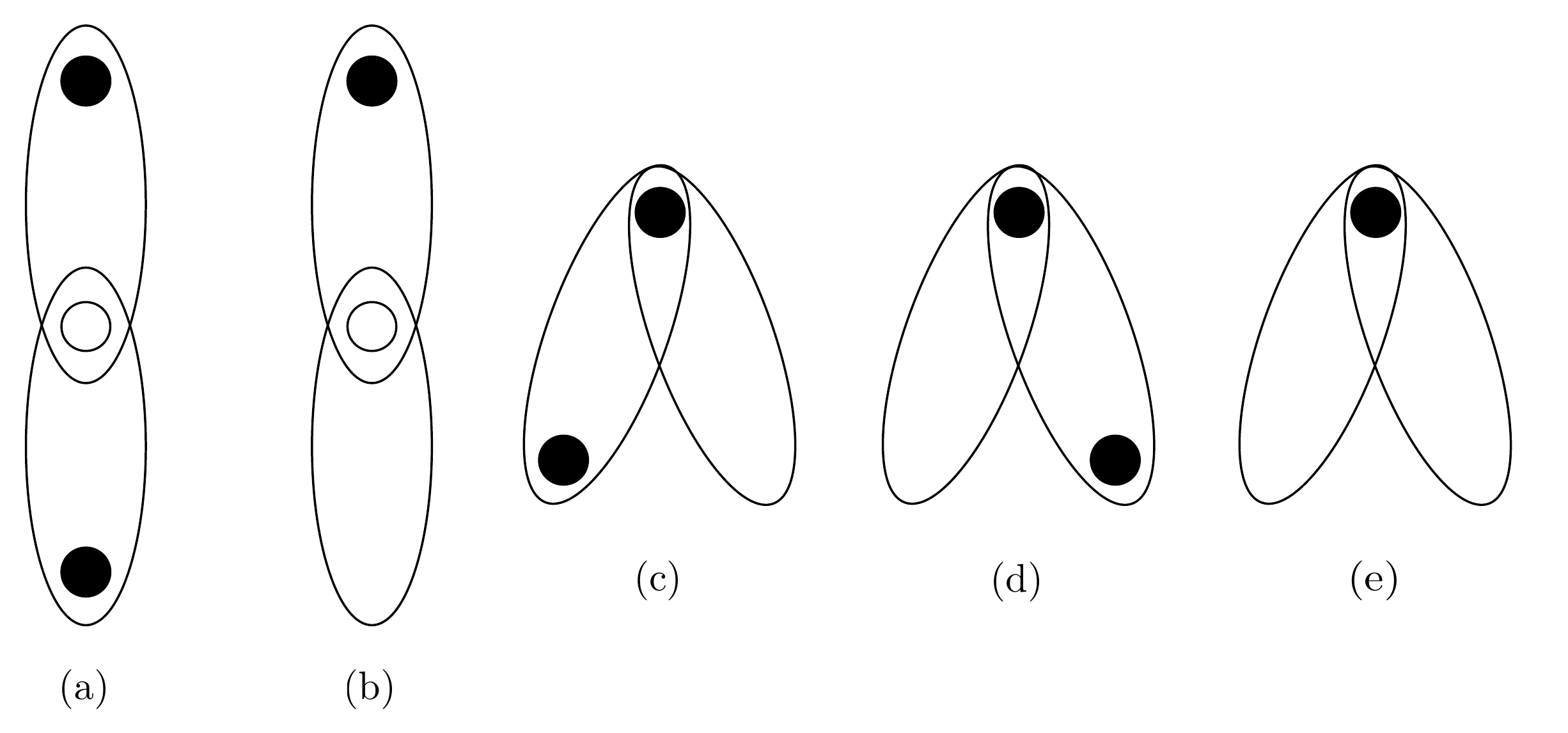}
   \caption{Five ways of merging clusters. The $\bullet$ nodes are boundary nodes that remain  boundary nodes in the merged cluster. The $\circ$ nodes are  boundary nodes that become internal (non-boundary) nodes in the merged cluster. Note that in the last four merges at least one of the  merged clusters has a top boundary node but no bottom boundary node.}
   \label{fig:merge}
\end{figure}

\subsection{Top Trees}
A \emph{top tree} $\T$ of $T$ is a hierarchical decomposition of $T$ into  clusters. It is 
an ordered, rooted, labeled, and binary tree defined as follows (see Figure~\ref{fig:toptrie}(a)-(c)).
\begin{itemize}
\item[$\bullet$] The nodes of $\T$ correspond to clusters of $T$.
\item[$\bullet$] The root of $\T$ corresponds to the cluster $T$ itself. The top boundary node of the root of $\T$ is the root of $T$. 
\item[$\bullet$] The leaves of $\T$ correspond to the edges of $T$. The label of each leaf is the label of the corresponding edge $(u,v)$ in $T$. 
\item[$\bullet$] Each internal node of $\T$ corresponds to the merged cluster of its two children. The label of each internal node is the type of merge it represents (out of the five merging options). The children are ordered so  that the left child is the child cluster visited first in a preorder traversal of $T$. 
\end{itemize}

\begin{lemma}[Alstrup et al.~\cite{AHLT2005}]
	Given a tree $T$ of size $n_T$, we can construct in $O(n_T)$ time a top tree $\T$ of $T$ that is of size $O(n_T)$ and height $O(\log n_T)$.
\end{lemma}
\begin{figure}[t]
   \centering
   \includegraphics[scale=0.25]{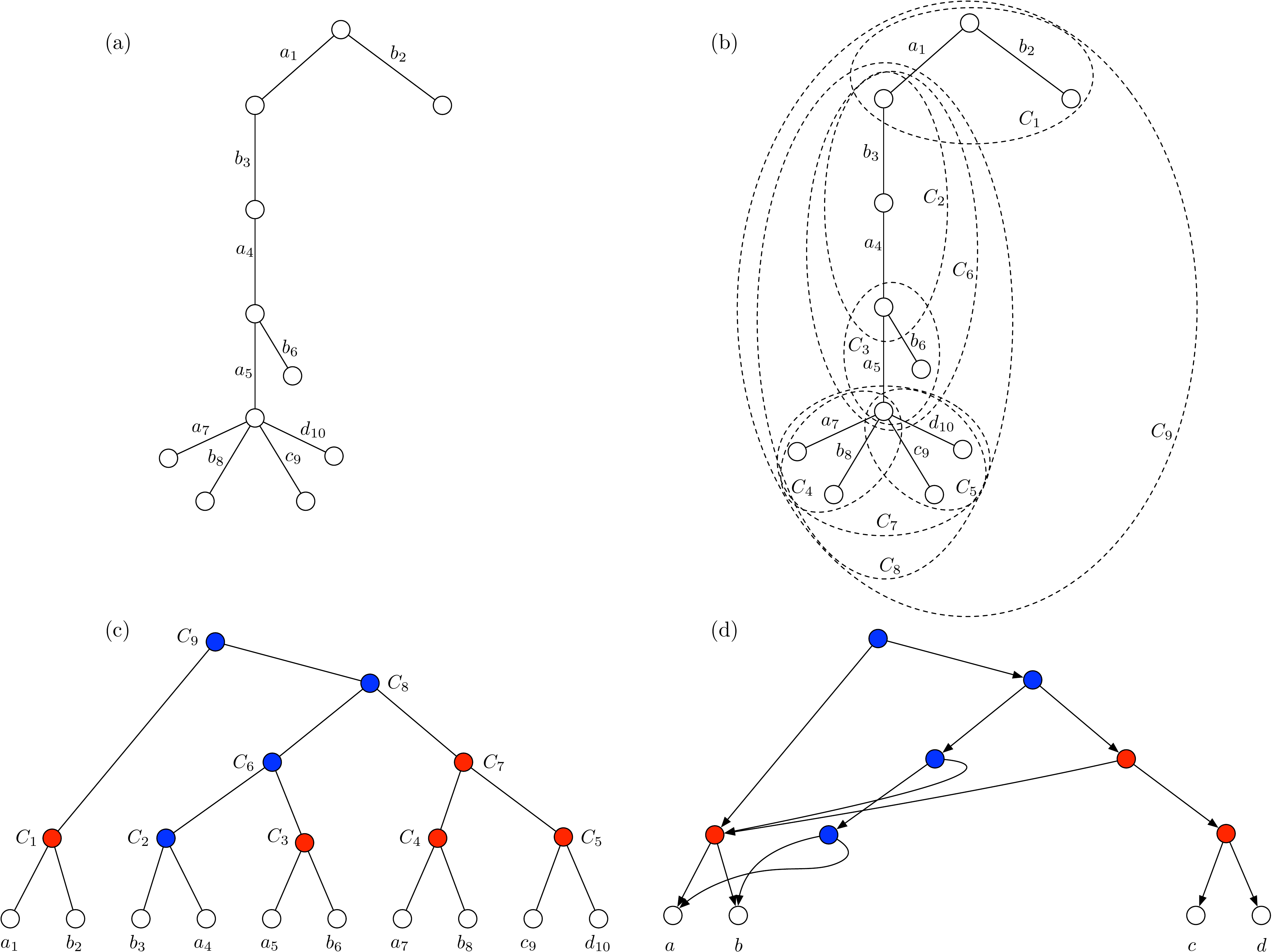}
   \caption{(a) A trie. Each edge label has a subscript to identify the corresponding leaf in the top tree in (c). 
   (b) A hierarchical clustering of (a). (c) The top tree corresponding to (a). Blue nodes are vertical clusters and red nodes are horizontal clusters.  (d) The top DAG of (c).}
   \label{fig:toptrie}
\end{figure}

\subsection{Top Dags} 
Every labeled tree can be represented with a directed acyclic graph (DAG) by identifying identical rooted subtrees and replacing them with a single copy.  
The \emph{top DAG} of $T$, denoted $\TD$, is the minimal DAG representation of the top tree $\T$ of $T$. We can compute it in $O(n_\T)$ time from $\T$~\cite{DST1980}\footnote{Here we use edge labels instead of nodes label. The two definitions are equivalent and edge labels are more natural for tries.}. Top DAGs have important properties for compression and computation~\cite{BGLW2015,BFG2017,HR2015,DG2018}. We need the following optimal worst-case  compression bound.

\begin{lemma}[Dudek and Gawrychowski~\cite{DG2018}] Given an ordered tree with $n_T$ nodes over an alphabet of size $\sigma$, we can construct a top DAG $\TD$ in $O(n_T)$ time of size $n_{\TD} = O(n_T/\log_{\sigma} n_T)$.
\end{lemma}

\section{A Simple Index}\label{sec:simple}
We first present a simple randomized Monte-Carlo word RAM string index, that will be the starting point for our deterministic, comparison-based pointer machine solution in the later sections.

\subsection{Data Structure}
Let $T$ be the trie of the strings $S = S_1, \ldots, S_k$ and let $\TD$ be the corresponding top DAG of $T$. Our data structure augments $\TD$ with additional information. For each cluster $C$ in $\TD$ we store the following information.
\begin{itemize}
	\item If $C$ is a leaf cluster representing an edge $e$, we store the label of $e$. 
	\item If $C$ is an internal cluster with left and right child $A$ and $B$, we store the label of the edge to the \emph{rightmost child} of the top boundary node, the fingerprint $\phi(\spine(C))$, and the length $|\spine(C)|$. 
\end{itemize}
This requires constant space for each cluster and hence $O(n_{\TD})$ space in total.

\subsection{Searching} 
Given a pattern $P$ of length $m$, we denote the unique node in $T$ whose path from the root matches the longest prefix of $P$, the

Given a pattern $P$ of length $m$ we find the longest matching prefix of $P$ in $T$, i.e., the unique node $\locus_{T}(P)$ in $T$ whose path from the root matches the longest prefix of $P$, as follows. First, compute and store all fingerprints of prefixes of $P$ in $O(m)$ time and space. By Lemma~\ref{lemma:fingerprints}, we can then compute the fingerprint of any substring of $P$ in $O(1)$ time.

Next, we traverse $\TD$ top-down while matching $P$. Initially, we search for $P[1, m]$ starting at the root of $\TD$. Suppose we have reached cluster $C$ and have matched $P[1,i]$. If $i = m$ we return $m$. Otherwise ($i < m$) there are three cases:
\begin{description}
\item[Case 1: $C$ is a leaf cluster.] Let $e$ be the edge stored in $C$. We compare $P[i+1]$ with the label of $e$. We return $i+1$ if they match and otherwise $i$. 

\item[Case 2: $C$ is a horizontal cluster.] Let $A$ and $B$ be the left and right child of $C$, respectively. We compare $P[i+1]$ with the label $\alpha$ of the edge to the rightmost child of $A$. If $P[i+1] \leq \alpha$, we continue the search in $A$ for $P[i+1\dots m]$. Otherwise, we continue the search in $B$ for $P[i+1\ldots m]$.  

\item[Case 3: $C$ is vertical cluster.] Let $A$ and $B$ be the left and right child of $C$, respectively. If $|\spine(A)| > m-i$ we continue the search in $A$ for $P[i+1 \ldots m]$. Otherwise, we compare the fingerprint $\phi(\spine(A))$ with $\phi(P[i+1 \ldots i+1 + |\spine(A)|])$. If they match, we continue the search in $B$ for $P[i +1 + |\spine(A)|\ldots m]$. Otherwise, we continue the search in $A$ for $P[i+1 \ldots m]$.
\end{description}

\begin{lemma}\label{lem:correctsimple}
	The algorithm correctly computes the longest matching prefix of $P$ in $T$.	 
\end{lemma}
\begin{proof}  We show by induction that at cluster $C$ the prefix $P[1,i]$ matches the path from the root of $T$ to $\topboundary(C)$ and $\locus_{T}(P)\in C$. If $C$ is the root of $\TD$ the empty path to $\topboundary(C)$ matches the empty prefix and $\locus_{T}(P) \in C = T$. Inductively, suppose $P[1,i]$ matches the path from the root to $\topboundary(C)$ and $\locus_{T}(P) \in C$. If $m = i$ the longest prefix is thus $P[1,m]$ and $\locus_T(P) = \topboundary(C)$. In each case, the algorithm maintains the invariant. The algorithm greedily matches as many characters from $P$ as possible, and hence at the end of the traversal the algorithm has found the longest matching prefix of $P$. 
\end{proof}

Next consider the running time. We compute all fingerprints of $P$ in $O(m)$ time. Each step of top-down traversal requires constant time and since the depth of $\TD$ is $O(\log n)$ the total time is $O(m + \log n)$. In summary, we have the following theorem.

\begin{theorem}\label{thm:simple}
	Let $S = S_1, \ldots, S_k$ be a set of strings of total length $n$, and let $\TD$ be the corresponding top DAG for the trie of $S$. On a word RAM model of computation, we can solve the compressed string indexing problem in $O(n_{\TD}) = O(n/\log_\sigma n)$ space and $O(m + \log n)$ time for any pattern of length $m$. The solution is randomized Monte-Carlo.   
\end{theorem}

In the next sections we show how to convert the above algorithm from a randomized algorithm on a word RAM machine into a deterministic algorithm on a pointer machine. We note that Theorem~\ref{thm:simple} and our subsequent solutions can be extended to other variants of prefix queries, such as \emph{counting queries}, that return the number of occurrences of $P$. To do so, we store the size of each cluster in $\TD$ and use the above top-down search modified to also record the highest cluster $E$ whose top boundary is $\locus_T(P)$. 
Since the size of $E$ is the number of occurrences of $P$, we obtain a solution that also supports counting within the same complexities. From $E$ we can also support \emph{reporting queries}, that return the strings in $S$ with prefix $P$, by simply decompressing $E$ incurring additional linear time in the lengths of the strings with matching prefix.

\section{Spine Extraction}\label{sec:spineextraction} 
We first consider how to handle vertical clusters (Case 3) deterministically on a pointer machine.  The key challenge is to efficiently extract the characters on the spine path of a vertical cluster from top to bottom without decompressing the whole cluster. We will use this to efficiently compute longest common prefixes between spine paths and substrings of $P$ in order to achieve total $O(m + \log n)$ time.  

Given the top DAG $\TD$, the \emph{spine path extraction problem} is to compactly represent $\TD$ such that given any vertical cluster $C$ we can return the characters of $\spine(C)$. We require that the characters are reported online and from top-to-bottom, that is, the characters must be reported in sequence and we can stop extraction at any point in time.  The goal is to obtain a solution that is efficient in the length of the reported prefix. In the following sections we show how to solve the problem in $O(n_{\TD})$ space and $O(m + \log n)$ total time over all spine path extractions. 

We present a new data structure derived from the top DAG called the \emph{vertical top DAG} and show how to use this to extract characters from a spine path. We then use this  to compute the longest common prefixes between a spine path and any string and plug this in to the top down traversal in the simple solution from Section~\ref{sec:simple} to obtain Theorem~\ref{thm:main}.

\subsection{Vertical Top Forest and Vertical Top DAG} 
\begin{figure}[t]
   \centering
   \includegraphics[scale=0.25]{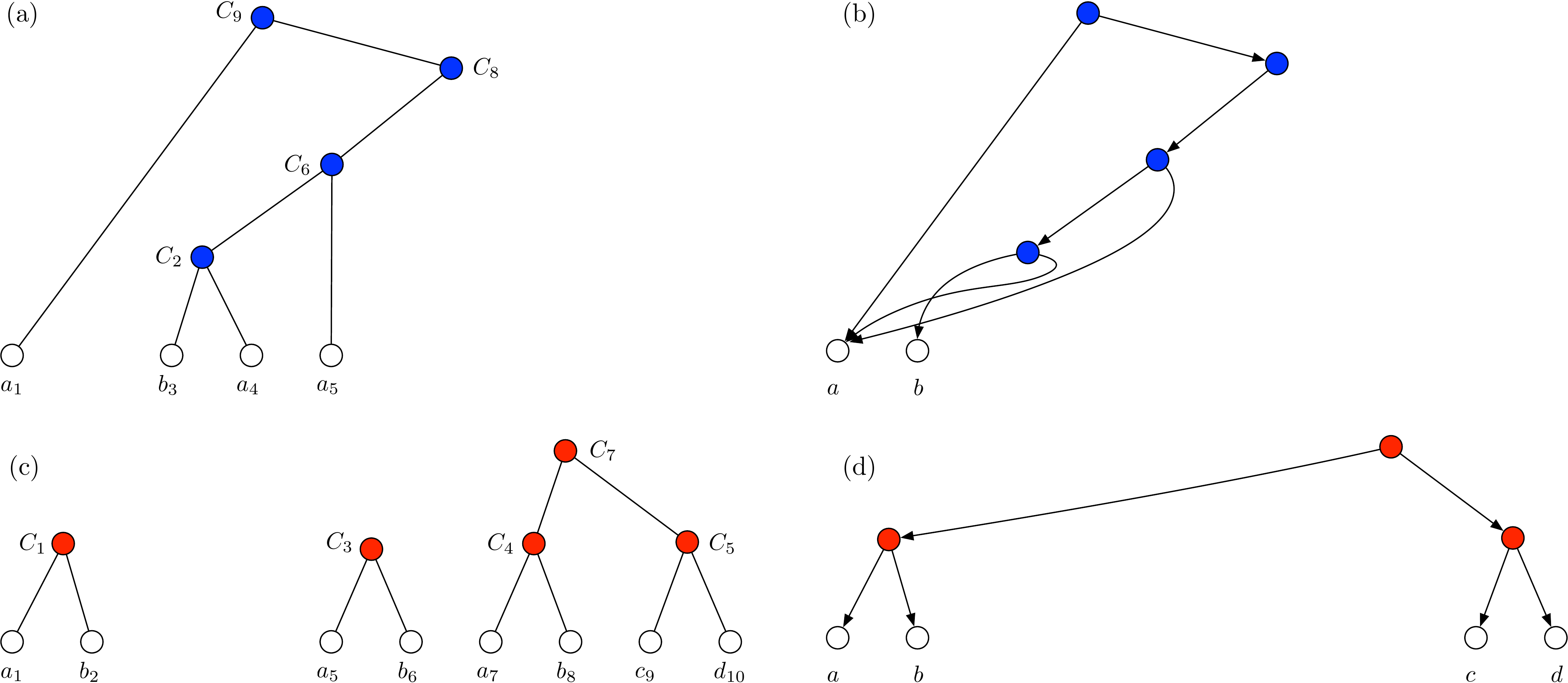}
   \caption{(a) The vertical top forest of the top tree from Figure~\ref{fig:toptrie}(c). (b) The vertical top DAG of (a). (c) The horizontal top forest of Figure~\ref{fig:toptrie}(a). (d) The horizontal top DAG of (a).}
   \label{fig:verticalhorizontal}
\end{figure}

The \emph{vertical top forest} $\V$ of $\T$ is a forest of ordered, rooted, and labeled binary trees. The nodes in $\V$ are all the vertical clusters of $\T$ and the leaf clusters of $\T$ that correspond to edges of a spine path of some cluster in $\T$. The edges of $\V$ are defined as follows. A cluster  $C$ of type (a) with children $A$ and $B$ in $\T$ has two children in $\V$. The left and right children are the unique vertical or leaf descendants of $C$ in $\T$ whose spine path is $\spine(A)$ and $\spine(B)$, respectively. A cluster  $C$ of type (b) with children $A$ and $B$ in $\T$ has a single child in $\V$, which is the unique vertical or leaf descendant of $C$ in $\T$ whose spine path is $\spine(A)$. 
See Figure~\ref{fig:verticalhorizontal}(a). We have the following correspondence between spine paths and subtrees in $\V$. 

\begin{lemma}\label{lem:verticaltopforest}
	Let $C$ be a vertical merge in $\V$ and $L$ be the leaves of $\V(C)$.  Then, $L$ are the edges on $\spine(C)$ and $|\V(C)| = O(|L|)$. Furthermore, the left-to-right ordering of $L$ corresponds to the top-down ordering of the edges on $\spine(C)$. 

\end{lemma}
\begin{proof}
By definition of $\V$ and the ordering of children in $\T$ and $\V$ it follows that the edges on the spine in top-down order are the leaves $L$ in left-to-right order. 
A cluster of type (b) in $\V(C)$ has a child that is either a leaf or a cluster of type (a). All clusters of type (a) have two children and hence $|\V(C)| = O(|L|)$.
\end{proof}

For instance in Figure~\ref{fig:verticalhorizontal}(a), the descendant leaves of $C_6$ are $b_3$, $a_4$, $a_5$ in left-to-right ordering corresponding to the edges in the spine of $C_6$ in Figure~\ref{fig:toptrie}(b).

The \emph{vertical top DAG} $\VD$ is the DAG obtained by merging identical subtrees of $\V$ according to the DAG compression of $\TD$. See Figure~\ref{fig:verticalhorizontal}(b). 


\subsection{Spine Extraction} 
We now show how to solve spine path extraction using the vertical top DAG $\VD$. The key idea is to simulate a depth-first left-to-right order traversal of $\V(C)$ using a recursive traversal of $\VD$. In order to use spine path extraction to search for a pattern we also need to be able to continue the search in some horizontal cluster of the top DAG after extracting characters on the spine. We will therefore define what we call a \emph{vertical exit cluster}, from which we can quickly find the cluster to continue the search from. 

%

Define the \emph{vertical exit cluster}, $\vexit(C,\ell)$, for $C$ at position $\ell$, $1< \ell \leq |\spine(C)|$ to be  the lowest common ancestor of leaves $\ell-1$ and $\ell$ in $\V(C)$. Intuitively, if we have extracted the first $\ell$ characters of $\spine(C)$, then  $\vexit(C,\ell)$ is the cluster such that all leaves in the left subtree have been extracted and only one leaf in the right subtree (corresponding to the $\ell$th character) has been extracted. Our goal is to implement spine path extraction in time $O(\ell + \height(C) - \height(\vexit(C,\ell)))$. This will yield a telescoping sum when doing multiple extractions.

Our data structure consists of the vertical top DAG $\VD$. We augment each internal cluster by the label of the first edge on its spine path and each leaf cluster by the label of the stored edge. This uses $O(n_{\VD})$ space. 

Given a cluster $C$ we implement spine path extraction by simulating a depth-first left-to-right order traversal of $\V(C)$ using a recursive traversal of $\VD$. To extract the first character we return the stored label at $C$. Suppose we have extracted $\ell -1$ characters,  $1<\ell\leq |\spine(C)|$. To extract the next character continue the simulated  depth-first search until we reach a cluster $D$ in $\V(C)$ whose leftmost leaf is the $\ell$th leaf of $\V(C)$. Return the character stored at $D$ and the parent of $D$ in $\V(C)$ as $\vexit(C,\ell)$. (Note the parent of $D$ is the cluster visited right before $D$ in the simulated depth-first search.) 

By Lemma~\ref{lem:verticaltopforest}, the algorithm correctly solves spine path extraction and the total time to extract $\ell$ characters is $O(\ell + \height(C) - \height(\vexit(C,\ell)))$. 
We need a stack to keep track of the current search path in the traversal using $O(\height(\V(C))) = O(\log n_\T) = O(n_\TD)$ space.  In summary, we have the following lemma.

\begin{lemma}\label{lem:spe-simple}
Let $\VD$ be the vertical top DAG. We can represent $\VD$ in $O(n_{\VD})$ space such that given a vertical cluster $C$, we can support spine path extraction on $C$ in $O(\ell + \height(C) - \height(\vexit(C,\ell)))$ time, where $\ell$ is the length of the extracted prefix of $\spine(C)$.
\end{lemma}

Note that we can use Lemma~\ref{lem:spe-simple} to compute the longest common prefix of $\spine(C)$ and any string by reporting the characters on the spine path from top-to-bottom and comparing them with the string until we get a mismatch. This uses $O(\ell +1 + \height(C) - \height(\vexit(C,\ell+1)))$ time, where $\ell$ is the length of the longest common prefix.

\section{An $O(m + \log n)$ Time Solution}\label{sec:part1puttingtogether}
We now plug in our spine path extraction algorithm from Section~\ref{sec:spineextraction} into the simple algorithm from Section~\ref{sec:simple}.  

Define the \emph{horizontal  entry cluster} for a vertical cluster $C$, denoted $\hentry(C)$, to be the highest horizontal cluster or leaf cluster in $\T(C)$ that contains all edges from $\topboundary(C)$ to children within $C$. For a horizontal cluster or a leaf the horizontal exit cluster is the cluster itself.
Note  $\hentry(C)$ is the highest horizontal cluster or leaf cluster on  the path from $C$ to the leftmost leaf of $C$.

Our data structure consists of the data structures from Section~\ref{sec:simple} without fingerprints and Section~\ref{sec:spineextraction}.
This uses $O(n_{\TD})$ space. To search for a string $P$ of length $m$, we use the same algorithm as in Section~\ref{sec:simple}, but with the following new implementation of the vertical merges.
\begin{description}
\item[Case 3: $C$ is vertical cluster.]  Recall we have reached a vertical cluster $C$ and have matched prefix $P[1, i]$.
We check if the first character on $\spine(C)$ matches $P[i+1]$. If it does not, we continue the algorithm from $\hentry(C)$.  If it does, we extract characters from $\spine(C)$ in order to compute the length $\ell$ of the longest common prefix of $\spine(C)$ and $P[i+1, m]$ and the corresponding vertical exit cluster $E = \vexit(C, \ell+1)$. Let $B$ be the right child of $E$ in $\TD$. We traverse the leftmost path from $B$ to find $\hentry(B)$ and continue the search for $P[i+\ell+1,m]$ from there.  
\end{description}

\begin{lemma}\label{lem:correctpluslog}
	The algorithm correctly computes the longest matching prefix of $P$ in $T$.		 
\end{lemma}
\begin{proof} 
	We show by induction that at cluster $C$ the prefix $P[1,i]$ matches the path from the root of $T$ to $\topboundary(C)$ and $\locus_{T}(P)\in C$. If $C$ is the root of $\TD$ the empty path to $\topboundary(C)$ matches the empty prefix and $\locus_{T}(P) \in C = T$. Inductively, suppose $P[1,i]$ matches the path from the root to $\topboundary(C)$ and $\locus_{T}(P) \in C$. If $m = i$ the longest prefix is thus $P[1,m]$ and $\locus_T(P) = \topboundary(C)$.
Correctness of Case 1 and Case 2 follows from Lemma~\ref{lem:correctsimple}.

Consider Case 3 and let $E$ and $B$ be as in the description. By induction and correctness of spine extraction it follows that $P[1,i+\ell]$ matches  the path from the root of $T$ to $\topboundary(B)$. By induction $\locus_{T}(P) \in C$ and thus $\locus_{T}(P)$ is a descendant of $\topboundary(B)$ in $C$.
Since $\topboundary(B)$ is not a boundary node in $E$ it follows that all ancestors of $B$ in $\TD$ contains exactly the same edges out of $\topboundary(B)$ as $B$. Hence, $\locus_{T}(P) \in B$.
\end{proof}

%
%
%

Consider the time used in a vertical step from a cluster $C$ . 
The time to compute the longest common prefix computation extracting $\ell$ characters and walking to the corresponding horizontal entry cluster $\hentry(\vexit(C,\ell))$ is $O(\ell + h(C) -  h(\vexit(C,\ell) + h(\vexit(C,\ell)) - h(\hentry(\vexit(C,\ell))) = O(\ell + h(C) -  h(\hentry(\vexit(C,\ell)))$.
Hence, if we have $z$ vertical steps from clusters $C_1,\ldots,C_z$ extracting $\ell_1,\ldots,\ell_z$ characters ending in $E_i=\hentry(\vexit(C_i,\ell_i))$, respectively, we use time
$$
\sum_{i=1}^z O(\ell_i + h(C_i) -  h(E_i)) = O\left(\sum_{i=1}^z \ell_i + h(C_1) - h(E_z)\right) = 
O(m+ \log n_\T)\;.
$$
This follows from the fact that $C_1, \ldots, C_z$ and $E_1, \ldots, E_z$ all lie on the same root-to-leaf path in $\T$ and that $h(E_i)\geq h(C_{i+1})$.
As in Section~\ref{sec:simple}, the total time used at horizontal merges is $O(\log n_T)$, as $E_1, \ldots, E_z$ all lie on the same root-to-leaf path in $\T$ and we only walk down in the tree during the horizontal merges.
This concludes the proof of the $O(m + \log n)$ query time in  Theorem~\ref{thm:main}.



\section{Spine Path Extraction with Constant Overhead}\label{sec:Spine Path Extraction}
Next, we show how to achieve the $O(m \log \sigma)$ query time in   Theorem~\ref{thm:main}. Our current solutions for horizontal merges (Case 2) from Section~\ref{sec:simple} and vertical merges (Case 3) from Section~\ref{sec:part1puttingtogether} both require $\Omega(m + \log n)$ and  hence we need new techniques for both cases to achieve the $O(m\log \sigma)$ time bound. We consider vertical merges in this section and horizontal merges in the next section.

In this section, we improve the total time used on spine extraction to optimal  $O(m)$ time. To do so we first introduce and present a novel solution to a new path extraction problem on trees in Section~\ref{sec:pathextractionintrees} and then show how to use this to extract characters from the spine in Section~\ref{sec:spineextractionoptimal}.

\subsection{Path Extraction in Trees}\label{sec:pathextractionintrees}
Given a tree $T$ with $n$ nodes, the \emph{path extraction problem} is to compactly represent  $T$ such that given a node $v$ we can return the nodes on the path from the root of $T$ to $v$ in constant time per node. We require that the nodes are reported online and from top-to-bottom, that is, the nodes must be reported in sequence and we can stop the extraction at any point in time. The ordering of the nodes from top to bottom is essential. The other direction (from $v$ to the root) is trivial since we can simply store parent pointers and traverse that path using linear space and constant time per node. If we allow word RAM tricks then we can easily solve the problem in the same bounds by using an existing \emph{level ancestor data structure} \cite{BFC2004, AH2000, Dietz1991}. We present an optimal solution that does not use word RAM tricks and works on a pointer machine. 
As mentioned in the introduction, an optimal solution can be also obtained by plugging in known tools, but we believe that our method is simpler
and elegant.

\begin{figure}[t]
   \centering
   \includegraphics[scale=0.4]{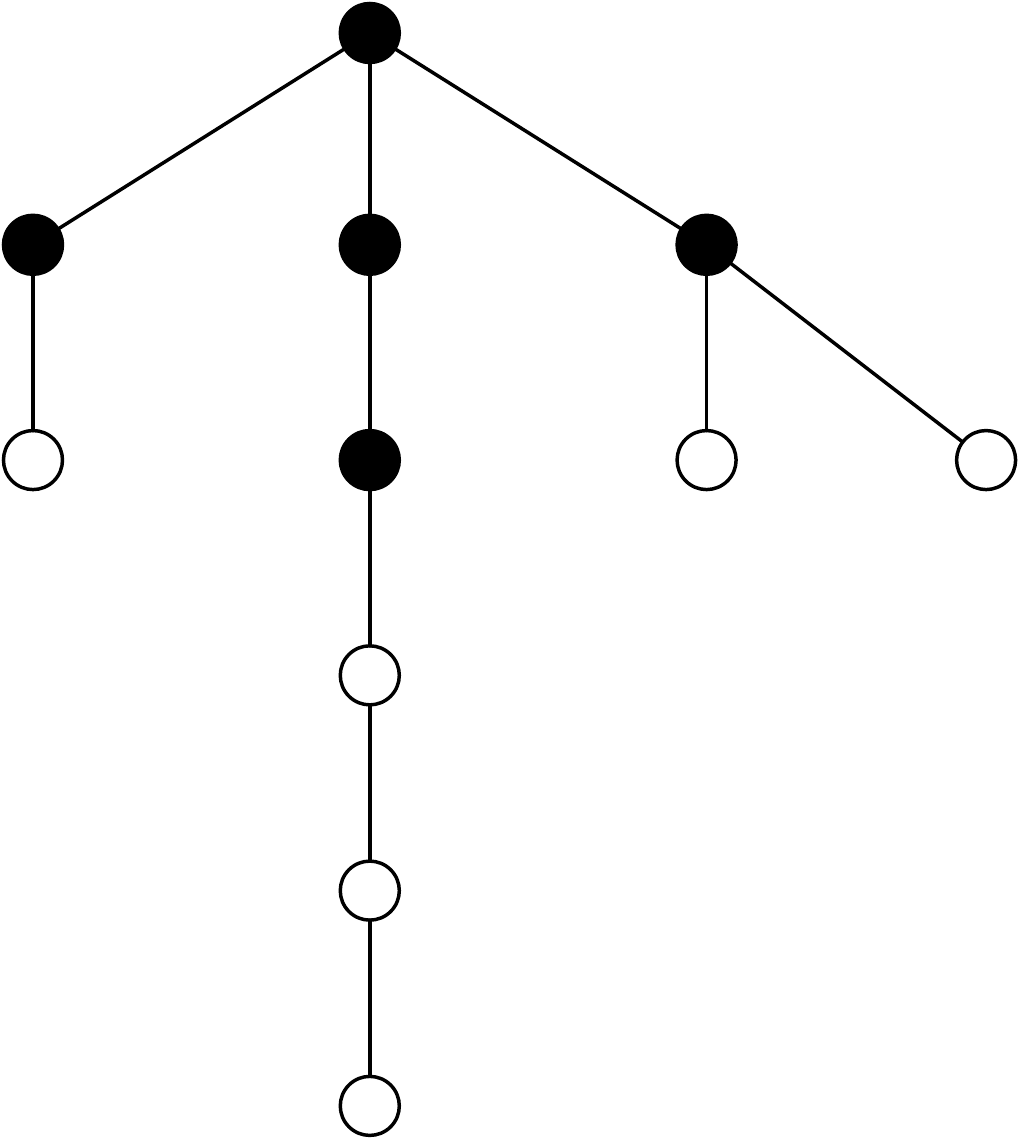}
   \caption{The tree decomposition for path extraction. The black nodes are the nodes in $T_\textrm{top}$ and the white nodes are the nodes in $T_\textrm{bot}$. The three root-to-leaf paths are stored as three linked lists sorted by increasing depth. The total size of the lists is $2 + 3 +2 = 7$. }
   \label{fig:pathextraction}
\end{figure}

Let $\depth(v)$ and $\height(v)$ be the distance from $v$ to the root and to deepest leaf in $v$'s subtree, respectively. Decompose $T$ into a top part $T_\textrm{top}$ consisting of nodes $v$, such that $\depth(v)\leq \height(v)$, and a bottom part $T_\textrm{bot}$ consisting of the remaining nodes. For each leaf $u$ in $T_\textrm{top}$ we store the path from the root of $T_\textrm{top}$ to $u$ explicitly in a linked list sorted by increasing depth. (see Figure~\ref{fig:pathextraction}). 
Note that multiple copies of the same node may be stored across different lists. Each such path to a leaf $u$ uses $O(\depth(u))$ space, and hence the total space for all paths in $T_\textrm{top}$ is $$\sum_{u \textrm{ a leaf in } T_\textrm{top}}\!\!\!\!\!\!\depth(u)  \ \ \ \leq \sum_{u \textrm{ a leaf in } T_\textrm{top}}\!\!\!\!\!\!\height(u) \ \ =\ \  O(n)\;,$$ 
where the first equality follows by definition of the decomposition and the second follows since the longest paths from a descendant leaf in $T(u)$ to a leaf $u$ in $T_\textrm{top}$ are disjoint for all the leaves $u$ in $T_\textrm{top}$. For all internal nodes in $T_\textrm{top}$ we store a pointer to a leaf below it. For all nodes $v$ in $T_\textrm{bot}$ we store a pointer to the unique ancestor $v$ that is a leaf in $T_\textrm{top}$. We answer a path extraction query for a node $v$ as follows. If $v$ is in $T_\textrm{top}$ we follow the leaf pointer and output the path  stored in this leaf from the root until we reach $v$. 
If $v$ is in $T_\textrm{bot}$ we jump to the unique ancestor leaf $u$ of $v$ in $T_\textrm{top}$.  We extract the path from the root to $u$, while simultaneously following parent pointers from $v$ until we reach $u$ storing these nodes on a stack. That is, each time we extract a node from the root-to-$u$ path we follow a parent pointer and put the next node on the stack. We stop pushing nodes to the stack when we reach $u$. When we have output all nodes from the root to the leaf in $T_\textrm{top}$ we output the nodes from the stack. Since $\depth(u) \leq \height(u)$ the path from the root to $u$ is at least as long as the path from $v$ to $u$ plus 1. Therefore, the whole path is extracted. We spend $O(1)$ time per node and hence we have the following result.

\begin{lemma}\label{lem:pathextraction}
Given a tree $T$ with $n$ nodes, we can solve the path extraction problem in linear space and preprocessing and constant time per reported node.
\end{lemma}

\subsection{Optimal Spine Path Extraction}\label{sec:spineextractionoptimal}
We plug the path extraction solution into our depth-first search traversal of the vertical top DAG $\VD$ to speed up spine extraction and longest common prefix computation. Recall that given a vertical cluster $C$, our goal is to simulate a depth-first left-to-right order traversal of the subtree $\V(C)$ using the vertical top DAG $\VD$. 

We construct the \emph{left-path suffix forest} $L$ of $\VD$ as follows. The nodes of $L$ are the nodes of $\VD$. If $C$ has a left child $A$ in $\VD$ then $A$ is the parent of $C$ in $L$. Hence, any leftmost path in $\VD$ corresponds to a path from a node to an ancestor of the node in $L$. We now store $L$ with the path extraction data structure from Lemma~\ref{lem:pathextraction}.  We implement the depth-first traversal as before except that whenever the traversal reaches an unexplored cluster $C'$ in $\V(C)$ we begin path extraction for that cluster corresponding to the path from $C'$ to the leftmost descendant leaf $\hat{C}$. We extract the leaf $\hat{C}$ and then continue the depth-first traversal from there. Hence, the current search path of the depth-first traversal is partitioned into an alternating sequence of leftmost paths and right edges. Whenever we need to go up on a left edge in the traversal we extract the next node for the corresponding path extraction instance. 

To extract the topmost $\ell$ characters of $\spine(C)$ we now use constant time to find the leftmost descendant leaf of $\V(C)$ and then $O(\ell)$ time to traverse the first $\ell$ leaves. Hence, we improve the time from $O(\height(\V(C)) + \ell)$ to $O(\ell)$.  At any point during the traversal we maintain ongoing path extractions instances along the current search path. The stacks each of these need are of size at most linear in the length of their corresponding subpath of the search path and hence this requires at most $O(\log n_\VD)$ extra space. 

\begin{lemma}\label{lem:spe-fast}
We can represent the vertical top DAG $\VD$ in $O(n_{\VD})$ space such that given  a vertical cluster $C$, we can support spine path extraction on $C$ in $O(\ell)$ time, where $\ell$ is the length of the extracted prefix of $\spine(C)$.
\end{lemma}

\section{Horizontal Access}\label{sec:horizontalaccess}
We now show how to efficiently handle horizontal merges (Case 2). In the simple algorithm from Section~\ref{sec:simple} we use constant time at each horizontal merge leading to an $O(\log n_\T)$ total time solution. Since we cannot afford $O(\log n_\T)$ time we instead show how to handle all horizontal merges in $O(m \log \sigma)$ time. The key idea is to convert the problem into a variant of the \emph{random access problem} for grammar compressed strings, and then design a  linear-space logarithmic-query solution to the random access problem. We describe the random access problem in Section~\ref{sec:randomaccess} and present our solution to it in Section~\ref{sec:Horizontal Access in Linear Space}, we introduce the horizontal top DAG in Section~\ref{sec:horizontaltopDAG}, and define and solve the \emph{horizontal access problem} in Section~\ref{sec:gappedgrammarshorizontalaccess}.

\vspace{-0.1in} \subsection{Grammars and Random Access}\label{sec:randomaccess}
Grammar-based compression replaces a long string $S$ by a small context-free grammar~(CFG)~$\G$. We view a grammar $\G$ as a DAG, where each node is a grammar symbol and each rule defines directed ordered edges from the righthand side to the lefthand side. Given a node $C$ in $\G$, we define $T(C)$ to be the parse tree rooted at $C$ and $S(C)$ to be the string consisting of the leaves of $T(C)$ in left-to-right order. Note that given a rule $C \rightarrow C_1 C_2 \ldots C_k$ we have that $S(C) = S(C_1) \cdot S(C_2) \cdots S(C_k)$, where $\cdot$ denotes concatenation. Given a grammar $\G$ representing a string $S$, the \emph{random access problem} is to compactly represent $\G$ while supporting fast \emph{access queries}, that is, given an index $i$ in $S$ report $S[i]$. Bille et al.~\cite{BLRSW2015} showed how to do random access in $O(\log |S|)$ time using $O(n_\G \cdot \alpha_k(n_\G))$ space\footnote{Here $\alpha_k(n)$ for any constant $k$ denotes the inverse of the $k^{th}$ row of Ackermann's function, defined
 as  $\alpha_k(n)=1+\alpha_k(\alpha_{k-1}(n))$ so that
  $\alpha_1(n)= n/2$, $\alpha_2(n)=\log n$, $\alpha_3(n)=\log^* n$, and so on.} on a pointer machine model. Furthermore, given a node $C$ in $\G$, access queries can be supported on the string $S(C)$ in time $O(\log |S(C)|)$. 

For our purposes, we need to slightly extend this result to \emph{gapped grammars}. A gapped grammar is a grammar except that each internal rule is now of the form $C \rightarrow C_1 g_1 C_2 \ldots g_{k-1} C_k$, where $g_i$ is a non-negative integer called the \emph{gap}. The string generated by $\G$ is now $S(C) = S(C_1)\texttt{0}^{g_1} S(C_2) \cdots S(C_{k-1}) \texttt{0}^{g_{k-1}} S(C_k)$ and hence the resulting string generated is as before except for the inserted gaps of runs of \texttt{0}'s. Note that $|S(C)| = |S(C_1)| + g_1 + |S(C_2)| + \cdots + g_{k-1} + |S(C_k|$. The above random access result is straightforward to generalize to gapped grammars:

\begin{lemma}[Bille et al.~\cite{BLRSW2015}]\label{lem:randomaccess} Let $S$ be a string compressed into a gapped grammar $\S$ of size $n_\S$. Given a node $v$ in $\S$, we can support random access queries in $S(v)$ in $O(\log (|S(v)|))$ time using $O(n_\S \cdot \alpha_k(n_\S))$ space. The solution works on a pointer machine model of computation. 
\end{lemma}

\subsection{Horizontal Access in Linear Space}\label{sec:Horizontal Access in Linear Space}
Bille et al.~\cite{BLRSW2015} further showed that the inverse-Ackermann factor in the space complexity of Lemma~\ref{lem:randomaccess} can be removed if we assume a word RAM model of computation. In this section we show that this can also be achieved on a pointer machine. To this end, we need to replace a single component in the solution of Bille et al.,  their \emph{weighted level ancestor} structure.
In the weighted level ancestor problem, we are given a tree $T$ on $n$ nodes with positive weights on the edges. For every node $u\in T$, let $d(u)$ be its distance to the root, and let $\parent(u)$ be its parent. Then, the goal is to preprocess $T$ to answer the following \emph{weighted level ancestor} queries: given a non-root node $u\in T$ and a positive number $x \leq d(u)$, find an ancestor $v$ such that $d(v) \geq x$ but $d(\parent(v)) < x$. 

Without getting into the proof of Lemma~\ref{lem:randomaccess}, it suffices to say that (1) performing a random access query boils down to performing $O(\log (|S(v)|))$ weighted level ancestor queries, and (2) in order for all these $O(\log (|S(v)|))$ queries to be done in total $O(\log (|S(v)|))$    time, the time for each weighted level ancestor query should be proportional to $\log\frac{d(u)}{d(v)-d(\parent(v))}$. Intuitively, we seek a position on an edge at distance $x$ from the root, and the longer the found edge is the smaller the query time should be. We next show how to achieve such query time using linear space on a pointer machine, implying an inverse-Ackermann factor improvement to  Lemma~\ref{lem:randomaccess}. 

\begin{lemma}
\label{lem:wla}
A tree $T$ on $n$ nodes can be preprocessed in $O(n)$ space to answer a weighted level ancestor query for a node $u\in T$ and a number
$x$ in $O(1+\log\frac{d(u)}{d(v)-d(\parent(v))})$ time, where $v$ is the found ancestor of $u$.
\end{lemma}

\begin{proof}
We start with partitioning $T$ into \emph{slices}. The $i^\text{th}$ slice, denoted $T_{i}$, consists of all nodes $u\in T$ such that
$d(u)\in [2^{i},2^{i+1})$. Observe that each $T_{i}$ is a collection of trees. For each node $u\in T_{i}$, we store a pointer to an
arbitrary descendant $v$ such that no child of $v$ belongs to $T_{i}$, denoted $\query(u)$ . In other words, $v$ is a leaf in its
corresponding tree of $T_{i}$ (and also a descendant of $u$ that belongs to the same tree of $T_{i}$). To answer a query for a
node $u\in T_{i}$ and a number $x$, we first replace $u$ with $\query(u)$. This does not increase $\log(d(u))$ by more than 1
and, because we replace $u$ with its descendant, returns the same node. Thus, from now on we can assume that the input
to a query is a node $u\in T_{i}$ that is a leaf in its tree of $T_{i}$. For each such node, we store a pointer $\jumpslice(u)$
to the highest ancestor of $u$ that still belongs to $T_{i}$. To answer a query for a node $u\in T_{i}$ that is a leaf
in its tree of $T_{i}$ and a number $x$, we then check the following three cases:
\begin{enumerate}
\item $x \leq d(\parent(\jumpslice(u)))$, then we repeat with $u$ replaced with $\parent(\jumpslice(u))$.
\item $x > d(\parent(\jumpslice(u)))$ and $x \leq d(\jumpslice(x))$, then we return $\jumpslice(x)$.
\item $x > d(\jumpslice(u))$, then we search for the answer among the ancestors of $u$ in its tree of $T_{i}$.
\end{enumerate}
Observe that whenever Case 1 applies the value of $\log(d(u))$ decreases by at least 1, and so it is enough to show
how to separately preprocess each tree of $T_{i}$ for weighted ancestor queries in $O(1+i-\log(d'(v)-d'(\parent(v))))$ time, where $v$
is the found node and $d'(v)$ is its distance to the root of the corresponding tree of $T_{i}$ (note that the maximum value of
$d'(v)$ is $2^{i}$).

We can therefore focus on the following problem: preprocess a tree $T$ with a parameter $i$ such that $d(u) \leq 2^{i}$ for every $u\in T$ for weighted ancestor queries
in $O(1+i-\log(d(v)-d(\parent(v))))$ time, where $v$ is the found ancestor of $u$, and $u$ is always a leaf. The preprocessing proceeds
recursively. We first partition $T$ into the top part, denoted $T_\text{top}$, and a collection of trees constituting the bottom
part, denoted $T_\text{bottom}$. A node $v\in T$ belongs to $T_\text{top}$ when $d(v) \leq 2^{i-1}$.  Each leaf
$u\in T_\text{bottom}$ stores a pointer $\checkptr(u)$ to its highest ancestor that still belongs to $T_\text{bottom}$.
Let $T_{\text{bottom}-}$ denote the collection of trees obtained by removing all leaves from $T_{\text{bottom}}$.
Each leaf $u\in T_\text{bottom}$ additionally stores a pointer $\jumptop(u)$ to an arbitrary leaf in the subtree rooted
at $\parent(\checkptr(u))$ in $T_\text{top}$, and a pointer $\jumpbottom(u)$ to an arbitrary leaf in the subtree rooted
at $\parent(u)$ in $T_{\text{bottom}-}$. We apply the above construction recursively with a parameter $(i-1)$ on $T_{\text{top}}$ and on every tree of $T_{\text{bottom}-}$.  See Figure~\ref{fig:alpha} for an illustration.

\begin{figure}[t]
\centering
\includegraphics[scale=0.85]{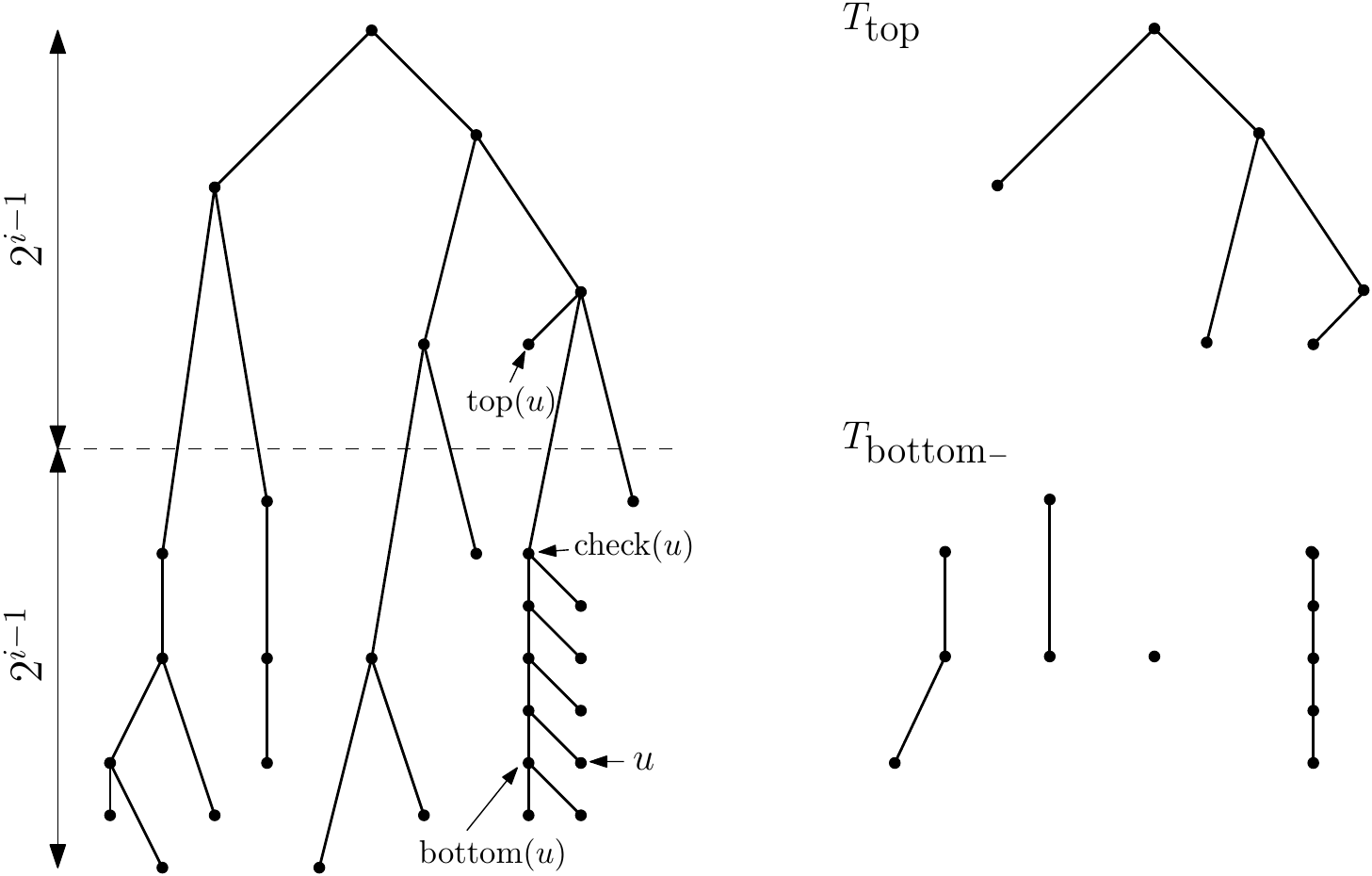}
\caption{Tree $T$ with parameter $i$ is decomposed into $T_{\text{top}}$ and $T_{\text{bottom}-}$.}
\label{fig:alpha}
\end{figure}

To answer a query for a leaf $u\in T_{\text{bottom}}$ and a number $x$, we check the following four cases:
\begin{enumerate}
\item $x \leq d(\parent(\checkptr(u)))$, then we repeat with $u$ replaced with $\jumptop(u)$ in $T_{\text{top}}$.
\item $x > d(\parent(\checkptr(u)))$ and $x\leq d(\checkptr(u))$, then we return $\checkptr(u)$.
\item $x > d(\parent(u))$, then we return $u$.
\item $x > d(\checkptr(u))$ and $x\leq d(\parent(u))$, then we repeat with $u$ replaced with $\jumpbottom(u)$ 
and $x$ decreased by $d(\checkptr(u))$ in the corresponding tree of $T_{\text{bottom}-}$.
\end{enumerate}
The cases are not mutually exclusive as it might happen that $\checkptr(u)=u$. Correctness of Case 2 and 3
is immediate. In Case 1 and 4 we recurse while maintaining the invariant that $u$ is a leaf in the current tree,
and the sought node is easily seen to belong to $T_{\text{top}}$ or $T_{\text{bottom}-}$ (because we require
$x\leq d(\parent(u))$ we can indeed consider $T_{\text{bottom}-}$ instead of $T_{\text{bottom}}$), respectively.
In every recursive step, the value of $i$ decreases by 1. Also, if $\lfloor\log(d(v)-d(\parent(v)))\rfloor=j$ then 
after $i-j+1$ steps the edge from $v$ to $\parent(v)$ cannot belong to the currently considered tree, and so
there are at most $i-j+1$ steps making the query time as required. To analyze the space, we assume that the partition
into $T_{\text{bottom}}$ and $T_{\text{top}}$ is only conceptual, and the stored information $\checkptr(u)$, $\jumptop(u)$
and $\jumpbottom(u)$ is associated with a node $u\in T$. Because the leaves of $T_{\text{bottom}}$ for which we need
to store information are then removed and do not participate further in the construction, this is indeed possible
and shows that the overall space is $O(1)$ per node of $T$. Finally, even though we have only described how
to answer a query for a leaf $u\in T_{\text{bottom}}$, the query algorithm rewritten to use the information stored
at nodes of $u\in T$ behaves as if $u\in T_{\text{bottom}}$ and hence is correct.
\end{proof}

\begin{corollary}\label{cor:randomaccess} Let $S$ be a string compressed into a gapped grammar $\S$ of size $n_\S$. Given a node $v$ in $\S$, we can support random access queries in $S(v)$ in $\log (|S(v)|)$ time using $O(n_\S)$ space. The solution works on a pointer machine model of computation. 
\end{corollary}

\subsection{Horizontal Top Tree and Horizontal Top DAGs}\label{sec:horizontaltopDAG}
Similar to the vertical top forest we define the \emph{horizontal top forest} $\H$ of $\T$ as a forest of ordered and rooted trees that consists of all horizontal clusters of $\T$ and leaves of $\T$ whose top boundary is shared with a horizontal cluster. We define the edges in of $C$ in $\H$ as follows. Let $C$ be a  horizontal cluster $C$ with children $A$ and $B$ in $\T$.  If $A$ is a horizontal cluster or a leaf then the left child of $C$ is $A$, and if $A$ is a vertical cluster then the left child of $C$ is $\hentry(A)$. Similarly, the right child of $C$ is either $B$ or $\hentry(B)$. See Figure~\ref{fig:verticalhorizontal}. 
We have the following property of~$\H$. 

\begin{lemma}\label{lem:horizontal}
	Let $C$ be a horizontal merge in $\H$. Then, the leaves of $\H(C)$ are the edges to children of the top boundary node of $C$ and the left-to-right ordering of the leaves correspond to the left-to-right ordering of the children of $C$ in $T$. All nodes in $\H(C)$ has $\topboundary(C)$ as top boundary node.
\end{lemma}
\begin{proof}
By definition of $\H$ and the ordering of the children in $\T$ and $\H$ it follows that the edges to children of the top boundary node of $C$ correspond to the leaves in $\H(C)$ in left-to-right order. Let $C$ be a horizontal cluster with children $A$ and $B$ in $\T$. Then $\topboundary(A) = \topboundary(B) = \topboundary(C)$. Furthermore, by definition $\topboundary(\hentry(C)) = \topboundary(C)$. Hence, all nodes in $\H(C)$ has $\topboundary(C)$ as top boundary node.
\end{proof}
For instance in Figure~\ref{fig:verticalhorizontal}(c) the descendant leaves of $C_7$ are $a_7$, $b_8$, $c_9$, and $d_{10}$ in left to right ordering corresponding to the edges to the children of $\topboundary(C_7)$. Given the horizontal top forest we define the \emph{horizontal top DAG} $\HD$ as the DAG obtained by merging the subtrees of $\H$ according to the DAG compression of $\T$ into $\TD$.

\subsection{Gapped Grammars and Horizontal Access}\label{sec:gappedgrammarshorizontalaccess}
Let $C$ be an internal cluster in $\H$. The \emph{spine child} of $C$ is the unique child of $C$ that contains the first edge of $\spine(C)$. 
A descendant cluster $D$ of $C$ is a  \emph{spine descendant} of $C$ if all clusters on the path from $C$ to $D$ are spine children of their parent. Define the \emph{horizontal exit cluster} for a horizontal cluster $C$ and character $\alpha$, denoted $\hexit(C,\alpha)$, to be the highest cluster in $\H(C)$ that has the unique leaf in $\H(C)$ labeled $\alpha$ as a spine descendant.

Given the horizontal top DAG $\HD$, the \emph{horizontal access problem}, is to compactly represent $\HD$ such that given a horizontal merge $C$ and a character $\alpha \in \Sigma$, we can efficiently determine if $\topboundary(C)$ has an edge to a child labeled $\alpha$ within $C$ and if so return the horizontal exit cluster $\hexit(C,\alpha)$. 
In this section, we show how to solve the horizontal access problem in $O(n_\HD)$ space and $O(\log \sigma)$ time. 

The \emph{characteristic vector} of a cluster $C$ is a binary string encoding the labels of edges to children of $\topboundary(C)$. More precisely, given a character $\alpha \in \Sigma$ define $\rank(\alpha) \in \{1, \ldots, \sigma\}$ as the rank of $\alpha$ in the sorted order of characters of $\Sigma$. Also, given a cluster $C$ in $\H$ define $\rank(C)$ to be the set of ranks of leaf labels in $\H(C)$. We define the characteristic vector $S(C)$ recursively as follows. If $C$ is a leaf cluster $S(C) = \texttt{1}$ and if $C$ is an internal cluster with children $C_1, \ldots, C_k$, then 
$
S(C) =	S(C_1)\texttt{0}^{g_1}S(C_2) \cdots S(C_{k-1})\texttt{0}^{g_{k-1}}S(C_k) ,
$
where $g_i = \min(\rank(C_{i+1})) - \max(\rank(C_{i})) + 1$. 
Note that $|S(C)| \leq \sigma$ for any cluster $C$. From the definition we have the following correspondence between the characteristic vector and the leaf labels of a cluster. 

\begin{lemma}\label{lem:characteristic} 
Given a cluster $C$ in $\H$ and a character $\alpha \in \Sigma$, $\alpha$ is a  leaf label in $\H(C)$ iff $S(C)[\rank(\alpha) - \min(\rank(C))] = \texttt{1}$.
\end{lemma}

Let $R_1, \ldots, R_z$ be the root clusters of the trees in $\H$ and note that if we add a virtual root cluster $R$ as the parent of $R_1, \ldots, R_z$, $\H$ is a gapped parse tree for the string $S = S(R_1) \cdots S(R_z)$.  Hence, the horizontal top DAG $\HD$ is a gapped grammar for the same string. By Lemma~\ref{lem:characteristic} we can determine if there is an edge labeled $\alpha$ out of $\topboundary(C)$ in $C$ using a random access query on the corresponding gapped grammar using time $O(\log |S(C)|) = O(\log \sigma)$. If this edge exists, we can also find $\hexit(C,\alpha)$ in the same time using similar ideas. More precisely, we have the following result. 

\begin{lemma}\label{lem:horizontalaccess}
Given a cluster $C$ in $\H$ and a character $\alpha \in \Sigma$ we can solve the horizontal acces problem in $O(n_{\HD})$ space and $O(\log \sigma)$ time. 
\end{lemma}

\begin{proof} 
By construction the characteristic vector of $S(C)$ has length at most $\sigma$. Hence, by  Corollary~\ref{cor:randomaccess}, we can determine if there is an edge $\alpha$ out of $\topboundary(C)$ in $C$ using $O(n_{\HD})$ space and $O(\log |S(C)|)= O(\log \sigma)$ time. If this is the case, we need to find $\hexit(C,\alpha)$ in the same complexity. To do so, we augment the random access result of Corollary~\ref{cor:randomaccess} as follows. 

We need the following definitions from Bille et al.~\cite{BLRSW2015} applied to $\HD$ to explain the approach. The \emph{heavy-path decomposition} of $\HD$ partitions $\HD$ into heavy and light edges with the property that any root-to-leaf path in $\HD$ is decomposed into an alternating sequence of $O(\log \sigma)$ heavy paths and single light edges. The \emph{heavy-path suffix forest} $F$ of $\HD$ compactly encodes the heavy paths of $\HD$ in $O(n_{\HD})$ space and has the property that a subpath of a heavy path in $\HD$ uniquely corresponds to a path from a node $v$ to an ancestor of $v$ in $F$. Our random access solution from Corollary~\ref{cor:randomaccess} on $\HD$ solves $O(\log \sigma)$ weighted ancestor queries on $F$ using Lemma~\ref{lem:wla} and computes the alternating sequence of heavy subpaths and single light edges from $C$ to the leaf cluster containing the edge labeled $\alpha$. 

We construct a new contracted forest $F'$ from $F$ as follows. Imagine we mark all the edges going to non-spine children in $\HD$.  Then, $\hexit(C,\alpha)$ is the highest descendant of $C$ whose path to the leaf containing $\alpha$ only consist of unmarked edges. Now mark the corresponding edges in $F$ and construct $F'$ be contracting all unmarked edges. The weight of a contracted node is the weight of the highest of its included nodes in $F$. A weighted ancestor query on $F'$ now identifies the node corresponding to the lowest horizontal entry cluster on the heavy path in $\HD$. Since we contract edges in $F$ and reweigh them by adding the contracted edges the time for the weighted ancestor query is no more than the time for the corresponding query in $F$. 

To find $\hexit(C,\alpha)$, we traverse the alternating sequence of heavy paths and light edges from top-to-bottom in $\HD$ to find the lowest marked edge whose lowest endpoint is $\hexit(C,\alpha)$. At each heavy path we use a weighted ancestor query and at each light edge we simply check if it is marked. In total, this takes $O(\log \sigma)$ time.   
\end{proof}

\section{An $O(m\log \sigma)$ Solution}\label{sec:mlogsigma}
We can now plug in the spine extraction from Section~\ref{sec:spineextractionoptimal} and the horizontal access from Section~\ref{sec:horizontalaccess} into the simple algorithm from Section~\ref{sec:simple}. 
Define the \emph{vertical entry cluster} for a horizontal cluster $C$, denoted $\ventry(C)$, to be the highest vertical cluster or leaf cluster in $\T(C)$ that contains the first edge on $\spine(C)$. 

Our data structure consists of the data structure from Section~\ref{sec:spineextractionoptimal} for spine path extraction and the data structure from Section~\ref{sec:horizontaltopDAG} for horizontal access. Furthermore, we store for each vertical cluster in $\TD$ a pointer to its horizontal entry cluster and for each horizontal cluster a pointer to its vertical entry cluster. In total this uses $O(n_\TD)$ space.

To search we alternate between horizontal accesses using Lemma~\ref{lem:horizontalaccess} and spine path extractions using Lemma~\ref{lem:spe-fast}. Instead of traversals to find entry clusters we jump directly using the new pointers.
Specifically, we have the following modified algorithm:

Initially, we search for $P[1, m]$ starting at the root of $\TD$. Suppose we have reached cluster $C$ and have matched $P[1,i]$. If $i = m$ we return $m$. Otherwise ($i < m$) there are three cases:
\begin{description}
\item[Case 1: $C$ is a leaf cluster.] Let $e$ be the edge stored in $C$. We compare $P[i+1]$ with the label of $e$. We return $i+1$ if they match and otherwise $i$. 

\item[Case 2: $C$ is a horizontal cluster.] Compute $E=\hexit(C,P[i+1])$. If $P[i+1]$ does not match return $i$. Otherwise, continue the search for $P[i+1,m]$ from $\ventry(E)$.

\item[Case 3: $C$ is vertical cluster.] We check if the first character on $\spine(C)$ matches $P[i+1]$. If it does not we continue the algorithm from $\hentry(C)$.  Otherwise, we extract characters from $\spine(C)$ in order to compute the length $\ell$ of the longest common prefix of $\spine(C)$ and $P[i+1, m]$ and the corresponding vertical exit cluster $E = \vexit(C, \ell+1)$. Continue the search for $P[\ell+1,m]$ from $\hentry(E)$.  
\end{description}

\begin{lemma}\label{lem:longestmatchingprefix}
		The algorithm correctly computes the longest matching prefix of $P$ in $T$.
\end{lemma}
\begin{proof} 
	We show by induction that at cluster $C$ the prefix $P[1,i]$ matches the path from the root of $T$ to $\topboundary(C)$ and $\locus_{T}(P)\in C$. If $C$ is the root of $\TD$ the empty path to $\topboundary(C)$ matches the empty prefix and $\locus_{T}(P) \in C = T$. Inductively, suppose $P[1,i]$ matches the path from the root to $\topboundary(C)$ and $\locus_{T}(P) \in C$. If $m = i$ the longest prefix is thus $P[1,m]$ and $\locus_T(P) = \topboundary(C)$.
Correctness of Case 1 and Case 3 follows from Lemma~\ref{lem:correctpluslog}.

Consider Case 2. There are two cases. If $P[i+1]$ does not match, then by induction $\locus_T(P) = \topboundary(C)$ and we are done.

Otherwise, $\topboundary(C)$ has an edge to a child $v$ labeled $P[i+1]$ in $C$ and $\locus_T(P)$ is a descendant of $v$. 
Let $E$ be as in the description. 
By Lemma~\ref{lem:horizontal} $\topboundary(E)= \topboundary(C)$ and by the definition of $\ventry(E)$ we have $\topboundary(\ventry(E)) = \topboundary(E)$. Hence, by induction $P[1,i]$ matches the path from the root of $T$ to $\topboundary(\ventry(E))$. 
Recall, that the horizontal exit cluster is the highest horizontal cluster in $C$ that has $P[i+1]$ as a spine descendant. Hence, every cluster on the path from $C$ to $E$ has $v$ and all descendants of $v$ in $C$ as internal nodes. In particular, $\locus_T(P)\in E$ and hence by definition $\locus_T(P)\in \ventry(E)$.  
\end{proof}

Consider the alternating sequence of horizontal accesses and spine extractions. Each time we go from a horizontal access to a spine extraction the current character of $P$ must match the first character on the spine. Hence, each horizontal access is on a distinct character of $P$ and the total number of horizontal accesses is at most $m$. By Lemma~\ref{lem:horizontalaccess} it follows that the total time for  horizontal accesses is  $O(m \log \sigma)$. Since the sequence is alternating the number of spine extractions is at most $m+1$. Hence, by Lemma~\ref{lem:spe-fast} the total time for spine extractions is at most $O(m)$. This concludes the proof of the $O(m \log \sigma)$ query time in  Theorem~\ref{thm:main}.

\section{Lower Bound}\label{sec:lowerbound}

In this section we prove Theorem~\ref{thm:comparison}. Namely, we show that any structure storing a set $S$ of strings of total length $n$ over
an alphabet of size $\sigma$ needs to perform $\Omega(\min(m+\log n,m\log\sigma))$ comparisons
to decide if a given pattern string $P[1,m]$ belongs to $S$. Every comparison should be of the form
``$P[i] \leq c$'', where $c$ is a character. Note that the size of the structure is irrelevant for us.
We start with a technical lemma that is the gist of our lower bound.

\begin{lemma}
\label{lem:comparison}
For any $\sigma\geq 2$ and $m$, any comparison-based algorithm that given a string $P[1,m]$ over
an alphabet of size $\sigma$ checks if $\sum_{i=1}^{m} P[i] = 0 \pmod 2$ needs to perform
$\Omega(m\log \sigma)$ comparisons in the worst case.
\end{lemma}

\begin{proof}
The number of strings $P[1,m]$ over an alphabet of size $\sigma$ such that $\sum_{i=1}^{m} P[i] = 0 \pmod 2$
is at least $\sigma^{m-1}\lfloor \sigma/2\rfloor \geq \sigma^{m}/4$. Consider the decision tree $T$
corresponding to a comparison-based algorithm that decides if $\sum_{i=1}^{m} P[i] = 0 \pmod 2$
using less than $m\log\sigma-2$ comparisons in the worst case. Each node of $T$ corresponds to a subset of
possible inputs of the form $[a_{1},b_{1}]\times \ldots \times [a_{m},b_{m}]$, in particular the root of $T$
corresponds to $[1,\sigma]\times\ldots\times [1,\sigma]$ and its leaves correspond to disjoint subsets
of inputs for which the answer is the same (yes or no) that together cover the whole $[1,\sigma]\times\ldots\times [1,\sigma]$.
Because the depth of $T$ is assumed to be less than $m\log\sigma-2$, $T$ contains less than
$2^{m\log\sigma-2}=\sigma^{m}/4$ leaves, so there exists a leaf corresponding to a subset of inputs
$[a_{1},b_{1}]\times \ldots \times [a_{m},b_{m}]$ and two distinct strings $P[1,m]$ and $Q[1,m]$ such that
$\sum_{i=1}^{m} P[i] = \sum_{i=1}^{m} Q[i] = 0 \pmod 2$ and $P[i],Q[i]\in [a_{i},b_{i}]$ for every $i=1,\ldots,m$.
Because $P[1,m]$ and $Q[1,m]$ are distinct, there exists $j$ such that $P[j]\neq Q[j]$, and without losing generality
$P[j] < Q[j]$. We define a new string $P'[1,m]$ by setting $P'[i]=P[i]$ for every $i\neq j$ and $P'[j]=P[j]+1$.
Then $\sum_{i=1}^{m} P'[i] = 1 + \sum_{i=1}^{m} P[i] = 1 \pmod 2$ and $P'[i] \in [a_{i},b_{i}]$ for every $i=1,\ldots,m$,
so the algorithm incorrectly decides that the answer for $P'[1,m]$ is the same as for $P[1,m]$.
\end{proof}

\noindent We proceed to the main part of the lower bound. Fix $\sigma \geq 2$, $n$ and $m\leq n$. We consider two cases.

\begin{description}
\item[$n\geq m\sigma^{m}$.] The set $S$ contains all strings $P[1,m]$ such that $\sum_{i=1}^{m} P[i] = 0 \pmod 2$.
There are at most $\sigma^{m}$ of such strings, and each of them is of length $m$, making their total
length at most $m\sigma^{m} \leq n$. Any structure that stores $S$ and allows checking if a given pattern $P[1,m]$
belongs to $S$ implies a comparison-based algorithm that checks if $\sum_{i=1}^{m} P[i] = 0 \pmod 2$.
By Lemma~\ref{lem:comparison}, this needs $\Omega(m\log \sigma)$ comparisons.

\item[$n < m\sigma^{m}$.] We choose the largest integer $\ell$ such that $n\geq m\sigma^{\ell}$ (by the assumption
on $n$ and $m\leq n$, $\ell \in [1,m)$). Now the set $S$ contains all strings $P[1,m]$ such that $\sum_{i=1}^{\ell} P[i] = 0 \pmod 2$
and $P[\ell+1]=\ldots=P[m]=\texttt{0}$. The total length of all such strings is at most $m\sigma^{\ell} \leq n$.
Any structure that stores $S$ and allows checking if a given pattern $P[1,m]$ belongs to $S$ implies a
comparison-based algorithm that checks if $\sum_{i=1}^{\ell} P[i] = 0 \pmod 2$ and additionally $P[\ell+1]=\ldots=P[m]=\texttt{0}$.
When executed with $P[1,m]=\texttt{0}^{m}$ the algorithm clearly needs to access every $P[i]$ and so perform
at least $m$ comparisons. Additionally, the algorithm can be converted into a procedure that given a pattern
$P[1,\ell]$ checks if $\sum_{i=1}^{\ell} P[i] = 0 \pmod 2$, which by Lemma~\ref{lem:comparison} requires
$\Omega(\ell\log \sigma)$ comparisons. Combining these two lower bounds we obtain that
$\Omega(m+\ell\log\sigma)$ comparisons are necessary. Rewriting the condition on $\ell$ and
using the assumption that $\ell \geq 1$, we obtain $\ell = \lfloor \log(n/m) / \log \sigma \rfloor \geq 1/2 \log(n/m) / \log \sigma$,
making our lower bound $\Omega(m+\log(n/m))=\Omega(m-\log m + \log n)=\Omega(m+\log n)$.
\end{description}

Combining the above two cases give us a lower bound of $\Omega(\min(m+\log n,m\log\sigma))$, because
depending on the value of $n$ we have a lower bound of either $\Omega(m\log \sigma)$ or $\Omega(m+\log n)$,
thus the minimum of these two is always a correct lower bound. This proves Theorem~\ref{thm:comparison}.

\bibliographystyle{plainurl}
\bibliography{paper}

\end{document}